\documentclass{article}

\usepackage{hyperref}
\usepackage{subcaption}

\usepackage{amsmath}
\usepackage{amsthm}
\newtheorem{thm}{Theorem}
\newtheorem{lem}[thm]{Lemma}

\usepackage{graphicx}
\usepackage{tikz}
\usetikzlibrary{positioning,arrows,patterns}

\usepackage{adjustbox}

\usepackage{siunitx}

\usepackage{latexsym}
\usepackage{algorithm}
\usepackage{algpseudocode}
\algnewcommand\algorithmicinput{\textbf{INPUT:}}
\algnewcommand\Input{\item[\algorithmicinput]}
\algnewcommand\algorithmicoutput{\textbf{OUTPUT:}}
\algnewcommand\Output{\item[\algorithmicoutput]}
\algnewcommand\algorithmicglobal{\textbf{GLOBAL VARIABLES:}}
\algnewcommand\Global{\item[\algorithmicglobal]}

\usepackage{color}

\usepackage{authblk}
\title{A Maximum Edge-Weight Clique Extraction Algorithm Based on Branch-and-Bound}
\author{Satoshi Shimizu\thanks{ss81054@gmail.com}}
\author{Kazuaki Yamaguchi\thanks{ky@kobe-u.ac.jp}}
\author{Sumio Masuda\thanks{masuda@kobe-u.ac.jp}}
\affil{Graduate School of Engineering, Kobe University}

\begin{document}
\maketitle

\begin{abstract}
The maximum edge-weight clique problem is
to find a clique whose sum of edge-weight is the maximum
for a given edge-weighted undirected graph.
The problem is NP-hard and some branch-and-bound algorithms have been proposed.
In this paper, we propose a new exact algorithm based on branch-and-bound.
It assigns edge-weights to vertices and calculates upper bounds using vertex coloring.
By some computational experiments, we confirmed our algorithm is faster
than previous algorithms.
\end{abstract}

\section{Introduction}
\label{sec:introduction}

For a simple undirected graph $G=(V,E)$,
a vertex subset $C\in V$ is called a clique if any pair of vertices in $C$ are adjacent.
Given a simple undirected graph $G=(V,E)$,
the maximum clique problem (MCP) is to find the clique of maximum cardinality.
MCP has lots of practical applications:
bioinformatics \cite{kc2002point},
coding theory \cite{etzion1998greedy,bogdanova2001error},
economics \cite{boginski2006mining} and more.
MCP is known to be NP-hard \cite{gary1979computers},
and the decision version is one of the Karp's 21 NP-complete problems \cite{karp1972reducibility}.
Since it has theoretical importance,
there have been a number of studies in decades \cite{wu2015survey}.

Given a simple undirected graph $G=(V,E)$ and non-negative weight $w(\cdot,\cdot)$ for each edge,
the maximum edge-weight clique problem (MEWCP) is to find the clique of maximum weight.
Obviously, MEWCP is a generalization of MCP.
Because of  edge-weights, MEWCP has practical applications that cannot be handled by MCP:
pattern recognition \cite{pavan2003generalizing}
protein side-chain packing \cite{bahadur2004protein,brown2006multiple},
market basket analysis \cite{cavique2007scalable},
communication analysis \cite{corman2002studying} and so on.

To obtain exact solutions of MEWCP,
there are two approaches in previous works.
One approach is formulating MEWCP into mathematical programming
and solve it by existing solvers.
Formulations based on
integer programming (IP) \cite{gouveia2015solving} and
mixed integer programming (MIP) \cite{shimizu2017mathematical}
were proposed.

The other approach is based on branch-and-bound.
Branch-and-bound algorithms recursively divide subproblems into smaller subproblems to search optimal solutions.
For each subproblem, it calculates upper bounds of the weight of feasible solutions
and prunes unnecessary subproblems that have no possibility to
improve the incumbent (current best solution).
Variety of algorithms adopt different strategies in
\textit{branching strategy}, \textit{search strategy} and \textit{pruning rule}.
The branching strategy is how to divide a given problem into subproblems.
The search strategy is the order in which subproblems are explored.
The pruning rule is how to calculate upper bounds to prune unnecessary subproblems.
A survey of branch-and-bound is shown in \cite{morrison2016branch}.
For MEWCP, CBQ proposed in \cite{hosseinian2018nonconvex} uses quadratic relaxation to obtain upper bounds.
Our previous algorithm EWCLIQUE is also based on the branch-and-bound \cite{shimizu2018branch}.
EWCLIQUE decomposes edge-weights of each subproblem into three components,
and calculates an upper bound for each of them.

In this paper, we propose a new branch-and-bound algorithm MECQ for MEWCP.
For each subproblem, our algorithm assigns weights of edges to vertices.
To obtain upper bounds using the assigned vertex weights,
our algorithm calculates vertex coloring that is a procedure
to divide the vertex set into a collection of independent sets.
By some computational experiments, we confirm our algorithm is faster than previous methods.

The remainder of this paper is organized as follows.
Our algorithm MECQ is described in Section \ref{sec:our algorithm}.
The results of computational experiments are in Section \ref{sec:experiments}.
We conclude the paper in Section \ref{sec:Conclusion}.

\section{Our algorithm MECQ}
\label{sec:our algorithm}

The proposed algorithm MECQ is based on the branch-and-bound.
Hereafter let $P(C,S)$ be a subproblem of MEWCP,
where $C$ is a constructed clique
and $S$ is a set of candidate vertices to be added to $C$.
Note that $C\subseteq N(v)$ must be satisfied for any element $v \in S$,
where $N(v)$ is the set of adjacent vertices of $v$.
$P(\emptyset,V)$ corresponds to the instance of MEWCP.
In this section, we first describe pruning rules of our algorithm.
Next, we show the branching strategy that divides $P(C,S)$ into subproblems,
and the search strategy to determine the order of subproblems to be solved.

\subsection{Pruning Rules}
\label{sec:Pruning Rules}

First, we describe an upper bound calculation for graphs where both vertices and edges are weighted.
Then, we show that an upper bound of $P(C,S)$ can be calculated in the same way.

For a graph $G=(V,E)$,
let $w(v)$ and $w(u,v)$ denote the weight of vertex $v$ and the weight of edge $(u,v)$, respectively.
Hereafter we define $w(u,v)=0$ for any $(u,v)\notin E$ for simplicity.
Let $G(S)$ be a subgraph of $G$ induced by a set of vertices $S\subseteq V$.
$E(S)$ denotes the edge set of $G(S)$.
For $S\subseteq V$, let $W(S)=\sum_{v\in S}w(v)+\sum_{(u,v)\in E(S)}w(u,v)$.

\subsubsection{Upper bound of vertex-and-edge-weighted graph}
\textit{Vertex coloring} is to color vertices
such that no adjacent vertices have the same color.
A vertex set of each color forms an independent set.
The smallest number of colors needed to color a graph $G$ is called chromatic number $\chi(G)$.
Let $\omega(G)$ be the clique number (the number of vertices in a maximum clique).
Since at most one vertex can be included in a clique from each independent set,
$\chi(G)$ is an upper bound of $\omega(G)$.
Therefore heuristic vertex coloring is often used to obtain upper bounds for MCP \cite{tomita2010simple,tomita2016much}.
For the maximum weight clique problem (MWCP),
the sum of the maximum vertex weight of each independent set is used as an upper bound
\cite{kumlander2004new,shimizu2012some}.
To calculate upper bounds of the MEWCP,
we consider assigning edge weights to incident vertices.
Let $\tau(v)$ be the index of the independent set including vertex $v$.
Namely, $\tau(v)=i$ for all $v\in I_i$.
Let $\sigma[v]$ be the total weight assigned to the vertex $v$ as follows:
\begin{equation}
    \sigma[v]=w(v) + \sum_{i<\tau(v)}\max\{w(u,v) \mid u\in I_i\cap N(v)\}. \label{eq:edge weight assignment}
\end{equation}

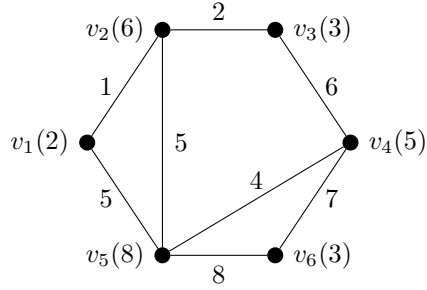
\begin{figure}[t]
    \centering
    \begin{adjustbox}{max width=\columnwidth}
        \begin{tikzpicture}[every node/.style={circle,draw=black,fill=black,inner sep=2pt}]
            \node[label=right:$v_6(3)$] (v6) at (2.5,0) {};
            \node[label=left:$v_5(8)$]  (v5) at (1,0) {};
            \node[label=right:$v_4(5)$] (v4) at (3.5,1.5) {};
            \node[label=right:$v_3(3)$] (v3) at (2.5,3) {};
            \node[label=left:$v_2(6)$]  (v2) at (1,3) {};
            \node[label=left:$v_1(2)$]  (v1) at (0,1.5) {};

            \draw[-] (v1) -- (v2) node[midway,left,draw=none,fill=none] {1};
            \draw[-] (v1) -- (v5) node[midway,left,draw=none,fill=none] {5};
            \draw[-] (v2) -- (v3) node[midway,above,draw=none,fill=none] {2};
            \draw[-] (v2) -- (v5) node[midway,right,draw=none,fill=none] {5};
            \draw[-] (v3) -- (v4) node[midway,right,draw=none,fill=none] {6};
            \draw[-] (v4) -- (v5) node[midway,above,draw=none,fill=none] {4};
            \draw[-] (v4) -- (v6) node[midway,right,draw=none,fill=none] {7};
            \draw[-] (v5) -- (v6) node[midway,below,draw=none,fill=none] {8};
        \end{tikzpicture}
    \end{adjustbox}
    \caption{a graph example $G_{ex}$}
    \label{fig:graph example}
\end{figure}
\begin{figure}[t]
    \centering
    \begin{tabular}{c|r|r|rl}
        && $\tau(\cdot)$ & \multicolumn{2}{c}{$\sigma[\cdot]$ of equation \ref{eq:edge weight assignment}} \\ \hline
        $I_1$&$v_1$ &     1    &2 & = $w(v_1)$ \\
        &$v_3$ &          &3 & = $w(v_3)$ \\
        &$v_6$ &          &3 & = $w(v_6)$ \\
        \hline
        $I_2$&$v_2$ &     2    &8 & = $w(v_2) + \max\{w(v_1,v_2),w(v_2,v_3)\}$ \\
        &$v_4$ &          &12 & = $w(v_4) + \max\{w(v_3,v_4),w(v_4,v_6)\}$ \\
        \hline
        $I_3$&$v_5$ &     3    &21 & = $w(v_5) + \max\{w(v_1,v_5),w(v_5,v_6)\}$\\
        &      &          &  &  \textcolor{white}{=}+ $\max\{w(v_2,v_5),w(v_4,v_5)\}$ \\
    \end{tabular}
    \caption{$\sigma[\cdot]$ for  $G_{ex}$}
    \label{fig:w_sigma for G_ex}
\end{figure}

An example $G_{ex}$ of a vertex-and-edge-weighted graph is shown in Figure \ref{fig:graph example}.
Numbers in parentheses are the vertex weights.
For $G_{ex}$, Figure \ref{fig:w_sigma for G_ex} shows
the assignment of independent set indices $\tau(\cdot)$ and
the weight $\sigma[\cdot]$ of equation \ref{eq:edge weight assignment}.
At most one vertex of $I_1\cap N(v_5) = \{v_1,v_6\}$
can be included in a clique since $I_1$ is an independent set.
Therefore we assign only the larger weight of edges $(v_1,v_5)$ and $(v_5,v_6)$ to $v_5$ for upper bound calculation.
We can ignore smaller weights to tighten upper bounds.

The following lemma shows that
an upper bound of the MEWCP can be calculated by using vertex coloring.

\begin{lem}
    \label{pr:lem upper bound}
    For a clique $C$ in a vertex-and-edge-weighted graph, the following inequality holds
    where $k=\max\{\tau(v) \mid v\in V\}$:
    \begin{eqnarray}
        W(C) \leq \sum_{i=1}^k\max\{\sigma[v]\mid v\in I_i\}. \label{eq:coloring upper bound}
    \end{eqnarray}
\end{lem}

\begin{proof}
    Since at most one vertex in $I_i$ can be included in $C$, $|C\cap I_i|\leq 1$ holds.
    From the definition, $C\subseteq N(v)$ for all $v\in C$.
    Therefore following inequality is obtained:
    \begin{eqnarray}
        W(C) & = & \sum_{v\in C} w(v) + \sum_{(u,v)\in E(C)}w(u,v)\\
        &  = & \sum_{v\in C} w(v) + \sum_{v\in C}\sum_{i<\tau(v)}\sum_{u\in C\cap I_i}w(u,v)\\
        &\leq& \sum_{v\in C} w(v) + \sum_{v\in C}\sum_{i<\tau(v)}\max\{w(u,v)\mid u\in N(v)\cap I_i\}\\
        &  = & \sum_{v\in C} \sigma[v]\\
        &  = & \sum_{i=1}^k\sum_{v\in C\cap I_i} \sigma[v]\\
        &\leq& \sum_{i=1}^k\max\{\sigma[v]\mid v\in I_i\}
    \end{eqnarray}
\end{proof}

Our algorithm uses Lemma \ref{pr:lem upper bound} to obtain upper bounds.
In the example $G_{ex}$, the clique of maximum weight is $\{v_4,v_5,v_6\}$ and its weight is 35.
Using $\sigma[\cdot]$ of Figure \ref{fig:w_sigma for G_ex},
an upper bound can be calculated by Lemma \ref{pr:lem upper bound} as follows:
\begin{eqnarray}
    W(C) & \leq & \max\{\sigma[v_1],\sigma[v_3],\sigma[v_6]\} + \max\{\sigma[v_2],\sigma[v_4]\} \nonumber\\
    & &  + \max\{\sigma[v_5]\} \\
    & = & 36
\end{eqnarray}

\subsubsection{Upper bound calculation for MEWCP}

Let $F$ be any feasible solution of a subproblem $P(C,S)$ of MEWCP.
$F$ is a union of $C$ and a subset of $S$.
$W(F)$ can be calculated as follows:
\begin{equation}
    W(F) =  W(C) + \sum_{u\in C}\sum_{v\in S\cap F}w(u,v) + W(S\cap F). \label{eq:three components}
\end{equation}
In the branch-and-bound, our algorithm calculates upper bounds of $W(F)$
to prune unnecessary subproblems.
Since the term $W(C)$ is obviously obtained in branching steps,
we have to calculate an upper bound of $\sum_{u\in C}\sum_{v\in S\cap F}w(u,v)$ $+$ $W(S\cap F)$.

To obtain vertex-and-edge-weighted graphs,
our algorithm assigns edge-weights of $\{(u,v)\mid u\in C,v\in S\cap F\}$
to vertices $v \in S\cap F$.
Let $w_\rho(C,v)$ be the total edge-weight assigned to $v \in S\cap F$ as follows:
\begin{equation}
    w_\rho(C,v) = \sum_{u\in C}w(u,v). \label{eq:pseudo vertex weight}
\end{equation}
Hereafter $w_\rho(v)$ denotes $w_\rho(C,v)$ when $C$ can be obviously identified.
Using  $w_\rho(v)$ and equation \ref{eq:three components},
following equation holds:
\begin{equation}
    W(F) =  W(C) + \sum_{v\in S\cap F}w_\rho(v) + W(S\cap F). \label{eq:three components using psuedo vertex weight}
\end{equation}
Note that the assigned weight $w_\rho(\cdot)$ and equation
\ref{eq:three components using psuedo vertex weight} is originally proposed in our previous work \cite{shimizu2018branch}.

For the vertex induced subgraph $G(S)$ of $P(C,S)$,
let $G(C,S)$ be the subgraph that can be obtained by assigning the weight $w_\rho(v)$ to each vertex in $G(S)$.
The proposed algorithm uses the vertex-and-edge-weighted graph $G(C,S)$ to calculate an upper bound of $W(F)$.
Equation \ref{eq:three components using psuedo vertex weight} indicates that
the sum of $W(C)$ and an upper bound of clique weight in $G(C,S)$ is an upper bound of $W(F)$.
Hence the proposed algorithm calculates an upper bound of $W(F)$ for $P(C,S)$ as follows:
\begin{enumerate}
    \item Obtain the vertex-and-edge-weighted graph $G(C,S)$
          using $w_\rho(\cdot)$ of equation \ref{eq:pseudo vertex weight}.
    \item Divide $S$ into mutually disjoint independent sets $I_1,I_2,\ldots,I_k$ by vertex coloring.
    \item Calculate $\sigma[\cdot]$ for each vertex in $G(C,S)$ using equation \ref{eq:edge weight assignment}.
    \item Calculate an upper bound of $W(F)$ using Lemma \ref{pr:lem upper bound}.
\end{enumerate}

\subsection{Branch-and-bound}
\label{sec:Branch and Bound}

Algorithm \ref{alg:MECQ} shows the main part of the proposed algorithm.
The inputs are a graph $G=(V,E)$, edge-weights $w(\cdot,\cdot)$
and an initial solution $C_{initial}$.
It searches for solutions by the recursive procedure EXPAND.
Our algorithm accepts a feasible solution $C_{initial}$ as an initial incumbent.
Although our algorithm works when $C_{initial}$ is empty,
given non-empty $C_{initial}$, our algorithm can use it as a lower bound and can efficiently prune subproblems in some cases.

\begin{algorithm}[!t]
    \caption{MECQ}
    \label{alg:MECQ}
    \begin{algorithmic}[1]
        \Input $G=(V,E)$, $w(\cdot,\cdot)$, $C_{initial}$
        \Output a maximum edge-weight clique $C_{max}$
        \Global $C_{max}$
        \State $C_{max}\leftarrow C_{initial}$
        \State \Call{expand}{$\emptyset,V$}
        \State \Return $C_{max}$
    \end{algorithmic}
\end{algorithm}

\begin{algorithm}[!t]
    \caption{Solving a subproblem}
    \label{alg:expand}
    \begin{algorithmic}[1]
        \Input a subproblem $P(C,S)$
        \Output Update $C_{max}$ to a better clique if it exists.
        \Global $C_{max}$

        \Procedure{expand}{$C,S$}
        \If{$S=\emptyset$} \label{step:start updating C}
        \If{$W(C) > W(C_{max})$} \label{step:updating C_max}
        \State $C_{max} \leftarrow C$
        \EndIf \label{step:updating C_max done}
        \State \Return \label{step:return without recursive calls}
        \EndIf \label{step:end updating C}
        \State $\Pi, upper[\cdot]\leftarrow$ \Call{CALC\_SEQ\_AND\_UB}{$C,S$}\label{step:call coloring}
        \For{each $p_i$ in order of $\Pi$} \label{step:expand for start} \Comment{$\Pi=[p_1,p_2,\ldots,p_{|S|}]$}
        \If{$W(C) + upper[p_i] > W(C_{max})$}\label{step:bounding}
        \State \Call{expand}{$C\cup\{p_i\},(S\setminus\{p_j\mid j<i\})\cap N(p_i)$} \label{step:recursive expand}
        \EndIf
        \EndFor \label{step:expand for end}
        \EndProcedure
    \end{algorithmic}
\end{algorithm}

Algorithm \ref{alg:expand} shows the recursive procedure EXPAND
to update the best solution so far.
When $S$ is empty, it is the base case that
updates the optimal solution $C_{max}$ (lines from \ref{step:start updating C} to \ref{step:end updating C}).
Otherwise, at line \ref{step:call coloring},
the function CALC\_SEQ\_AND\_UB returns
a sequence $\Pi=[p_1,p_2,\ldots,p_{|S|}]$ of vertices in $S$
and an array $upper[\cdot]$ of upper bounds using vertex coloring (described in \ref{sec:subroutine coloring}).
In the loop of lines from \ref{step:expand for start} to \ref{step:expand for end} in Algorithm \ref{alg:expand},
it recursively searches solutions at line \ref{step:recursive expand}.
The branching strategy, pruning rules and search strategy of our algorithm are as follows:
\begin{description}
    \item[Branching Strategy] \mbox{}\\
        For each $p_i$ of $\Pi$,
        our algorithm generates a child subproblem $P(C\cup \{p_i\},(S\setminus\{p_j\mid j<i\})\cap N(p_i))$.
        Excepting the order of vertices in $\Pi$,
        this strategy is same as previous algorithm EWCLIQUE \cite{shimizu2018branch}
        and is widely used in branch-and-bound algorithms of MCP and MWCP \cite{san2011exact,tomita2007efficient,shimizu2017fast,fang2016exact}.
    \item[Pruning Rules] \mbox{}\\
        For each $P(C\cup \{p_i\},(S\setminus\{p_j\mid j<i\})\cap N(p_i))$,
        an upper bound based on equation \ref{eq:coloring upper bound} is stored in the array $upper[p_i]$.
    \item[Search Strategy] \mbox{}\\
        In order of $\Pi=[p_1,p_2,\ldots p_{|S|}]$,
        our algorithm searches $P(C\cup \{p_i\},(S\setminus\{p_j\mid j<i\})\cap N(p_i))$.
        Since $\Pi=[p_1,p_2,\ldots,p_{|S|}]$ is ordered in non-increasing of $upper[\cdot]$ (described in \ref{sec:subroutine coloring}),
        this strategy is to find cliques of large weight early.
\end{description}

\subsubsection{Subroutine CALC\_SEQ\_AND\_UB}
\label{sec:subroutine coloring}

Algorithm \ref{alg:Vertex coloring} shows the function CALC\_SEQ\_AND\_UB.
It receives a subproblem $P(C,S)$ and returns a sequence $\Pi=[p_1,p_2,\ldots,p_{|S|}]$ of vertices in $S$ and an array $upper[\cdot]$ of upper bounds.
The array $upper[p_i]$ contains an upper bound of $P(C\cup \{p_i\},(S\setminus\{p_j\mid j<i\})\cap N(p_i))$.
It is used at line \ref{step:bounding} of Algorithm \ref{alg:expand}.
$\Pi$ is ordered in non-increasing of $upper[\cdot]$ and is used in branching strategy and search strategy.

\begin{algorithm}[tb]
    \caption{Calculate a vertex sequence and upper bounds}
    \label{alg:Vertex coloring}
    \begin{algorithmic}[1]
        \Input a subproblem $P(C,S)$
        \Output a vertex sequence $\Pi$ and an array $upper[\cdot]$

        \Procedure{CALC\_SEQ\_AND\_UB}{$C,S$}
        \For {$v\in S$}
        \State $\sigma[v]\leftarrow w_\rho(v)$ \label{step:init w_sigma}
        \EndFor
        \State $S'\leftarrow S$ \Comment{uncolored vertex set}
        \State $k\leftarrow 0$ \Comment{number of independent sets}
        \While {$S'\ne \emptyset$} \label{step:coloring loop start}
        \State $k\leftarrow k+1$
        \State $I_k\leftarrow \emptyset$
        \State $X\leftarrow S'$ \Comment{candidate vertex set to add to $I_k$}
        \While {$X\ne \emptyset$} \label{step:create maximal independent set start}
        \State $v\leftarrow$ a vertex of minimum $\sigma[\cdot]$ in $X$ \label{step:select minimum weight vertex}
        \State $upper[v]\leftarrow \sigma[v] + \sum_{i < k} \max\{\sigma[u]\mid u\in I_i\}$
        \State $I_k\leftarrow I_k \cup \{v\}$ \Comment{$\tau(v)=k$}
        \State Append $v$ to the head of $\Pi$. \label{step:add to seq}
        \State $X\leftarrow X\setminus N(v)$
        \State $S'\leftarrow S'\setminus \{v\}$
        \EndWhile \label{step:create maximal independent set end}
        \For {$v\in S'$}
        \State $\sigma[v]\leftarrow \sigma[v] + \max\{w(u,v)\mid u\in N(v)\cap I_k\}$ \label{step:add weight}
        \EndFor
        \EndWhile \label{step:coloring loop end}

        \State \Return $\Pi,upper[\cdot]$

        \EndProcedure
    \end{algorithmic}
\end{algorithm}

Here we describe the detail of Algorithm \ref{alg:Vertex coloring}.
At line \ref{step:init w_sigma},
it initializes $\sigma[\cdot]$ to $w_\rho(\cdot)$.
Each iteration of the while loop from line
\ref{step:coloring loop start} to \ref{step:coloring loop end},
it increments $k$ and constructs a maximal independent set $I_k$,
appends the vertices in $I_k$ to $\Pi$, and updates $\sigma[\cdot]$.
The loop terminates when all vertices are added to $\Pi$.
In the loop of lines from \ref{step:create maximal independent set start}
to \ref{step:create maximal independent set end},
it constructs a maximal independent set.
In line \ref{step:add to seq},
it appends the vertices in $I_k$ to the head of $\Pi$ in order of assignment to independent sets.
In the maximal independent set construction,
$X$ is the set of candidate vertices to be added to the independent set.
At line \ref{step:select minimum weight vertex},
our algorithm picks vertices from $X$ in non-decreasing order of $\sigma[\cdot]$.
This makes $\Pi$ non-increasing order of upper bounds.
At line \ref{step:add weight},
our algorithm updates $\sigma[\cdot]$ for vertices that
are not added to any independent set and are adjacent to vertices in the constructed independent set.

\section{Computational experiments}
\label{sec:experiments}

We implemented our algorithm MECQ in C++ to compare with previous algorithms.
In the experiments, our algorithm received an initial solution $C_{initial}$
calculated by phased local search (PLS) \cite{pullan2008approximating}.
PLS is a heuristic based on local search.
To avoid to be trapped into local optimums,
it switches three phases that have different search policies.
The one iteration of PLS consists of
50 searches of random phase,
50 searches of penalty phase and
100 searches of degree phase.
Our algorithm used PLS with 10 iterations and used
the best solution found by the PLS as an initial solution.

\subsection{Random graphs}

We generated uniform random graphs.
Edge-weights were uniform random integer values from 1 to 10.
The compared algorithms are EWCLIQUE \cite{shimizu2018branch}
and mathematical programming formulations of MIP proposed in \cite{shimizu2017mathematical}.
We used the C++ implementation of EWCLIQUE that was used in our previous work \cite{shimizu2018branch}.
For the formulations of MIP, we used the mathematical programming solver IBM CPLEX 12.5.

The compiler is g++ 5.4.0 with optimization option \mbox{-O2}.
The OS is Linux 4.4.0.
The CPU is Intel\textregistered Core\texttrademark i7-6700 CPU 3.40 GHz.
RAM is 16GB.
Note that CPLEX is a multi-thread solver based on branch-and-cut,
and our algorithm is a single-thread solver based on branch-and-bound.

Table \ref{tab:CPU time for random graphs} shows the CPU time for random graphs.
The symbol $\epsilon$ shows that the CPU time is less than 0.01 sec.
The column LB shows the weight of initial solutions given by PLS.
For all conditions,
our algorithm MECQ obtained optimal solutions in a shorter time than previous methods.
For the random graphs, the initial solution given by PLS does not improve performance.

Table \ref{tab:iterations for random graphs} shows the number of recursive iterations of MECQ and EWCLIQUE.
The value of the time [\si{\micro}s] is calculated by CPU time per iteration.
The value of the iteration ratio is the ratio of iterations of MECQ and EWCLIQUE.
From the result, we confirm that
although the computation time of upper bounds of MECQ is longer than EWCLIQUE,
the iterations of MECQ is less than our previous algorithm EWCLIQUE.
The difference of CPU time can be explained by this.
One reason is MECQ calculates upper bounds of equation \ref{eq:three components using psuedo vertex weight}
at once.
EWCLIQUE calculates upper bounds in two steps and calculates the sum of two upper bounds.

\begin{table*}[tb]
    \caption{CPU time for random graphs [sec]}
    \label{tab:CPU time for random graphs}
    \centering
    \begin{adjustbox}{max width=\textwidth}
        \begin{tabular}{rrr|r|rrr|r|rr}
            \hline
            &     &   \multicolumn{1}{c|}{optimal}     &     \multicolumn{4}{c|}{MECQ + PLS}      &        MECQ &       EWCLIQUE &       MIP  \\
            $|V|$ & $d$ &  \multicolumn{1}{c|}{weight} & LB &  PLS & MECQ &  Total & without PLS & \multicolumn{1}{c}{\cite{shimizu2018branch}} & \multicolumn{1}{c}{\cite{shimizu2017mathematical}} \\ \hline
              300 & 0.1 &   60.7 &   60.7 & 0.01 & $\epsilon$ &   0.01 & $\epsilon$ & $\epsilon$ &   91.54  \\
              350 & 0.1 &   64.8 &   64.8 & 0.01 & $\epsilon$ &   0.01 & $\epsilon$ & $\epsilon$ &  190.96  \\
            15000 & 0.1 &  174.7 &  148.3 & 0.73 &     408.56 & 409.29 &     402.75 &     460.90 & $>$1000  \\
              250 & 0.2 &   97.3 &   97.3 & 0.02 & $\epsilon$ &   0.02 & $\epsilon$ & $\epsilon$ &   64.26  \\
              280 & 0.2 &  102.4 &  102.4 & 0.02 & $\epsilon$ &   0.02 & $\epsilon$ & $\epsilon$ &  119.47  \\
             5500 & 0.2 &  254.8 &  212.2 & 0.37 &     319.27 & 319.64 &     319.93 &     440.29 & $>$1000  \\
              200 & 0.3 &  150.0 &  150.0 & 0.02 & $\epsilon$ &   0.02 & $\epsilon$ & $\epsilon$ &   39.16  \\
              250 & 0.3 &  155.5 &  155.5 & 0.03 &       0.01 &   0.03 &       0.01 &       0.01 &   97.16  \\
             2500 & 0.3 &  332.8 &  291.1 & 0.26 &     227.70 & 227.96 &     232.99 &     459.05 & $>$1000  \\
              160 & 0.4 &  185.5 &  185.5 & 0.02 &       0.01 &   0.03 &       0.01 &       0.01 &   21.98  \\
              200 & 0.4 &  224.0 &  224.0 & 0.03 &       0.01 &   0.04 &       0.01 &       0.02 &   57.54  \\
             1400 & 0.4 &  444.3 &  406.8 & 0.21 &     293.71 & 293.92 &     295.09 &     758.22 & $>$1000  \\
              140 & 0.5 &  272.7 &  272.7 & 0.01 &       0.01 &   0.02 &       0.01 &       0.02 &   21.28  \\
              170 & 0.5 &  300.6 &  300.6 & 0.03 &       0.03 &   0.06 &       0.03 &       0.06 &   52.72  \\
              750 & 0.5 &  560.3 &  546.5 & 0.15 &     164.32 & 164.47 &     164.76 &     603.91 & $>$1000  \\
              120 & 0.6 &  399.0 &  399.0 & 0.01 &       0.02 &   0.03 &       0.02 &       0.05 &   18.10  \\
              130 & 0.6 &  424.6 &  424.6 & 0.01 &       0.03 &   0.04 &       0.03 &       0.07 &   28.57  \\
              450 & 0.6 &  754.2 &  745.9 & 0.03 &     125.43 & 125.46 &     123.59 &     716.48 & $>$1000  \\
              100 & 0.7 &  583.5 &  583.5 & 0.01 &       0.03 &   0.04 &       0.04 &       0.11 &   15.12  \\
              110 & 0.7 &  607.1 &  607.1 & 0.01 &       0.06 &   0.07 &       0.06 &       0.24 &   31.23  \\
              270 & 0.7 & 1049.7 & 1049.1 & 0.02 &      61.97 &  61.99 &      62.78 &     589.15 & $>$1000  \\
               80 & 0.8 &  879.0 &  879.0 & 0.01 &       0.04 &   0.05 &       0.05 &       0.16 &    7.28  \\
               90 & 0.8 &  978.0 &  978.0 & 0.01 &       0.11 &   0.12 &       0.12 &       0.44 &   21.51  \\
              170 & 0.8 & 1580.2 & 1580.2 & 0.01 &      35.82 &  35.83 &      37.13 &     485.50 & $>$1000  \\
               70 & 0.9 & 1708.4 & 1708.4 & 0.01 &       0.09 &   0.10 &       0.11 &       0.62 &    3.80  \\
               80 & 0.9 & 2059.2 & 2059.2 & 0.01 &       0.35 &   0.37 &       0.37 &       2.93 &   14.84  \\
              110 & 0.9 & 2666.4 & 2666.4 & 0.02 &      22.40 &  22.41 &      23.27 &     590.83 & $>$1000  \\ \hline
        \end{tabular}
    \end{adjustbox}
\end{table*}

\begin{table*}[tb]
\caption{Iterations for random graphs}
\label{tab:iterations for random graphs}
\centering
\begin{adjustbox}{max width=\textwidth}
    \begin{tabular}{rr|r|rr|rr|r}
        \hline
        &     & \multicolumn{3}{c|}{MECQ} & \multicolumn{2}{c|}{EWCLIQUE} & $\frac{\mbox{MECQ}}{\mbox{EWCLIQUE}}$ \\
        & & with PLS & \multicolumn{2}{c|}{without PLS} & & & \multicolumn{1}{c}{iteration} \\
        $|V|$ & $d$ & iterations & iterations & time [\si{\micro}s] & iterations & time [\si{\micro}s] & \multicolumn{1}{c}{ratio} \\ \hline
          300 & 0.1 &      312.6 &      357.2 & $\epsilon$ &       1835.1 & $\epsilon$ & 19.46\% \\
          350 & 0.1 &      443.4 &      514.5 & $\epsilon$ &       2721.0 & $\epsilon$ & 18.91\% \\
        15000 & 0.1 & 56824211.6 & 56842808.1 &       7.09 &  702255007.1 &       0.66 &  8.09\% \\
          250 & 0.2 &      974.5 &     1043.8 & $\epsilon$ &       6412.8 & $\epsilon$ & 16.28\% \\
          280 & 0.2 &     1424.7 &     1515.2 & $\epsilon$ &       9862.9 & $\epsilon$ & 15.36\% \\
         5500 & 0.2 & 75843118.8 & 75882002.3 &       4.22 & 1537843711.0 &       0.29 &  4.93\% \\
          200 & 0.3 &     1547.8 &     1644.2 & $\epsilon$ &      14169.8 & $\epsilon$ & 11.60\% \\
          250 & 0.3 &     3449.5 &     3896.9 &       2.57 &      34844.2 &       0.29 & 11.18\% \\
         2500 & 0.3 & 65558818.2 & 65637225.6 &       3.55 & 1824084938.8 &       0.25 &  3.60\% \\
          160 & 0.4 &     2740.6 &     2895.5 &       3.45 &      31813.3 &       0.31 &  9.10\% \\
          200 & 0.4 &     5062.9 &     5868.7 &       1.70 &      70701.9 &       0.28 &  8.30\% \\
         1400 & 0.4 & 73127999.6 & 73361911.2 &       4.02 & 2993314273.1 &       0.25 &  2.45\% \\
          140 & 0.5 &     4541.5 &     5200.1 &       1.92 &      88105.7 &       0.23 &  5.90\% \\
          170 & 0.5 &    11374.6 &    11829.1 &       2.54 &     224608.1 &       0.27 &  5.27\% \\
          750 & 0.5 & 37702959.2 & 38847817.4 &       4.24 & 2342210511.1 &       0.26 &  1.66\% \\
          120 & 0.6 &     8166.3 &     8804.9 &       2.27 &     208388.4 &       0.24 &  4.23\% \\
          130 & 0.6 &    11338.8 &    13157.6 &       2.28 &     288494.1 &       0.24 &  4.56\% \\
          450 & 0.6 & 27505132.9 & 28469833.9 &       4.34 & 2725875895.6 &       0.26 &  1.04\% \\
          100 & 0.7 &    12737.1 &    13792.0 &       2.90 &     437892.1 &       0.25 &  3.15\% \\
          110 & 0.7 &    24203.6 &    26032.4 &       2.30 &     950029.7 &       0.25 &  2.74\% \\
          270 & 0.7 & 13547235.5 & 14604499.2 &       4.30 & 2141882035.0 &       0.28 &  0.68\% \\
           80 & 0.8 &    17659.5 &    20083.0 &       2.49 &     617626.3 &       0.26 &  3.25\% \\
           90 & 0.8 &    37616.7 &    45298.8 &       2.65 &    1578193.4 &       0.28 &  2.87\% \\
          170 & 0.8 &  7947844.6 &  8715999.5 &       4.26 & 1510657832.0 &       0.32 &  0.58\% \\
           70 & 0.9 &    30957.5 &    37263.2 &       2.95 &    2355972.7 &       0.26 &  1.58\% \\
           80 & 0.9 &   102852.5 &   111080.3 &       3.33 &    9974393.5 &       0.29 &  1.11\% \\
          110 & 0.9 &  5009165.9 &  5402618.4 &       4.31 & 1951189872.0 &       0.30 &  0.28\% \\ \hline
    \end{tabular}
\end{adjustbox}
\end{table*}

\subsection{DIMACS benchmarks}
\label{sec:DIMACS benchmarks}

DIMACS is a set of benchmarks for MCP \cite{dimacs}.
We used them as benchmarks of MEWCP by giving weights to edges in the same way as
\cite{shimizu2018branch,pullan2008approximating,gouveia2015solving,hosseinian2018nonconvex}.
For each edge $(v_i,v_j)$, we gave the weight $w(v_i,v_j)=(i+j) \mbox{ mod } 200+1$.

For the DIMACS benchmarks,
the results of computational experiments for previous methods are shown in \cite{hosseinian2018nonconvex,gouveia2015solving}.
Hence we also compared our algorithm with the branch-and-bound algorithm CBQ \cite{hosseinian2018nonconvex}
and mathematical programming formulations proposed in \cite{gouveia2015solving}.
We quote the results shown in \cite{hosseinian2018nonconvex,gouveia2015solving} to our result tables.
The CPU used in \cite{hosseinian2018nonconvex} is Intel\textregistered Core\texttrademark i7 2.90 GHz.
The CPU used in \cite{gouveia2015solving} is Intel\textregistered Core\texttrademark i7 3.40 GHz.

Table \ref{tab:CPU time for DIMACS} shows the CPU time for DIMACS.
Table \ref{tab:iterations for DIMACS} shows the number of recursive iterations of MECQ and EWCLIQUE.
Except for hamming8-2 and san200\_0.9\_1,
our algorithm MECQ obtained optimal solutions in a shorter time than others.
For hamming8-2 and san200\_0.9\_1, MECQ has usable performance.
Only the MECQ with PLS solved all instances in the table in 1000 sec.
Although the initial solutions given by PLS did not improve performance in random graphs,
they worked well in DIMACS.
Especially for benchmark families gen and san, it reduced a lot of computation time.

\begin{table*}[tb]
    \caption{CPU time for DIMACS (sec)}
    \label{tab:CPU time for DIMACS}
    \centering
    \begin{adjustbox}{max width=\textwidth}
        \begin{tabular}{rrrr|r|rrr|r|rrrrr}
            \hline
            &      &      &   \multicolumn{1}{c|}{optimal}         &  \multicolumn{4}{c|}{MECQ + PLS}  &           MECQ &  EWCLIQUE &       \multicolumn{1}{c}{MIP} &        \multicolumn{1}{c}{CBQ} & \multicolumn{1}{c}{G\&M}  & \multicolumn{1}{c}{IP$_{base}$}\\
            &  $|V|$ &  $d$ & \multicolumn{1}{c|}{weight} & LB &  PLS & MECQ &  Total & without PLS & \multicolumn{1}{c}{\cite{shimizu2018branch}} & \multicolumn{1}{c}{\cite{shimizu2017mathematical}} & \multicolumn{1}{c}{\cite{hosseinian2018nonconvex}} & \multicolumn{1}{c}{\cite{gouveia2015solving}} & \multicolumn{1}{c}{\cite{gouveia2015solving,hosseinian2018nonconvex}}  \\ \hline
                 brock200\_1 &  200 & 0.75 &  21230 &  21230 &       0.02 &      24.81 &      24.82 &      24.12 &     338.31 & $>$1000 &   3047.565 & $>$10800 & $>$10800 \\
                 brock200\_2 &  200 & 0.50 &   6542 &   6542 &       0.04 &       0.04 &       0.08 &       0.06 &       0.10 &  109.66 &      7.436 &  9464.24 & $>$10800 \\
                 brock200\_3 &  200 & 0.61 &  10303 &  10303 &       0.01 &       0.40 &       0.41 &       0.42 &       1.27 &  743.58 &     55.905 & $>$10800 & $>$10800 \\
                 brock200\_4 &  200 & 0.66 &  13967 &  13967 &       0.01 &       1.16 &       1.17 &       1.17 &       4.84 & $>$1000 &    188.031 & $>$10800 & $>$10800 \\
                      C125.9 &  125 & 0.90 &  66248 &  66248 &       0.02 &      24.21 &      24.23 &      24.83 &    $>$1000 & $>$1000 &   4558.170 & $>$10800 & $>$10800 \\
                  c-fat200-1 &  200 & 0.08 &   7734 &   7734 &       0.01 & $\epsilon$ &       0.01 & $\epsilon$ & $\epsilon$ &    4.80 &      0.483 &    3.870 &   31.296 \\
                  c-fat200-2 &  200 & 0.16 &  26389 &  26389 &       0.02 & $\epsilon$ &       0.02 & $\epsilon$ & $\epsilon$ &    4.72 &      0.890 &   33.260 &   49.671 \\
                  c-fat200-5 &  200 & 0.43 & 168200 & 168200 &       0.07 & $\epsilon$ &       0.07 & $\epsilon$ &      74.31 &    7.06 &   $>$10800 &  155.300 &  134.578 \\
                 c-fat500-10 &  500 & 0.37 & 804000 & 804000 &       0.31 &       0.23 &       0.54 &       0.24 &    $>$1000 &  745.93 &            &          &          \\
                  c-fat500-1 &  500 & 0.04 &  10738 &  10738 &       0.01 & $\epsilon$ &       0.01 & $\epsilon$ & $\epsilon$ &  171.91 &            &          &          \\
                  c-fat500-2 &  500 & 0.07 &  38350 &  38350 &       0.02 & $\epsilon$ &       0.03 & $\epsilon$ & $\epsilon$ &  399.90 &            &          &          \\
                  c-fat500-5 &  500 & 0.19 & 205864 & 205864 &       0.10 &       0.01 &       0.11 &       0.01 &       0.43 &  264.44 &            &          &          \\
                  DSJC500\_5 &  500 & 1.00 &   9626 &   9626 &       0.03 &      10.14 &      10.17 &      10.00 &      44.43 & $>$1000 &            &          &          \\
            gen200\_p0.9\_55 &  200 & 0.90 & 150839 & 150839 &       0.02 &     236.94 &     236.96 &    $>$1000 &    $>$1000 & $>$1000 &            &          &          \\
                  hamming6-2 &   64 & 0.90 &  32736 &  32736 &       0.01 & $\epsilon$ &       0.01 & $\epsilon$ & $\epsilon$ &    0.07 &      4.437 &    0.300 &   17.000 \\
                  hamming6-4 &   64 & 0.35 &    396 &    396 &       0.00 & $\epsilon$ & $\epsilon$ & $\epsilon$ & $\epsilon$ &    0.22 &      0.031 &    1.970 &    6.468 \\
                  hamming8-2 &  256 & 0.97 & 800624 & 800624 &       0.05 &      20.63 &      20.67 &      20.34 &       0.23 &    7.80 &   $>$10800 & $>$10800 & $>$10800 \\
                  hamming8-4 &  256 & 0.64 &  12360 &  12360 &       0.01 &       0.54 &       0.55 &       0.55 &       1.46 &  276.15 &    439.437 & $>$10800 & $>$10800 \\
               johnson16-2-4 &  120 & 0.76 &   3808 &   3766 & $\epsilon$ &       0.17 &       0.17 &       0.18 &       0.25 &   57.40 &     84.687 & $>$10800 & $>$10800 \\
                johnson8-2-4 &   28 & 0.56 &    192 &    192 & $\epsilon$ & $\epsilon$ & $\epsilon$ & $\epsilon$ & $\epsilon$ &    0.03 & $\epsilon$ &    0.140 &    0.421 \\
                johnson8-4-4 &   70 & 0.77 &   6552 &   6552 & $\epsilon$ & $\epsilon$ & $\epsilon$ & $\epsilon$ & $\epsilon$ &    0.40 &      0.687 &    2.340 &   65.171 \\
                     keller4 &  171 & 0.65 &   6745 &   6745 &       0.02 &       0.20 &       0.22 &       0.21 &       0.70 &  167.84 &     42.218 & $>$10800 & $>$10800 \\
                    MANN\_a9 &   45 & 0.93 &   5460 &   5460 &       0.01 &       0.02 &       0.03 &       0.02 &       0.02 &    1.22 &      1.906 &    9.390 &  130.344 \\
                p\_hat1000-1 & 1000 & 0.24 &   5436 &   5253 &       0.11 &       1.91 &       2.02 &       1.97 &       2.92 & $>$1000 &            &          &          \\
                p\_hat1500-1 & 1500 & 0.25 &   7135 &   6875 &       0.18 &      19.28 &      19.46 &      19.65 &      32.73 & $>$1000 &            &          &          \\
                 p\_hat300-1 &  300 & 0.24 &   3321 &   3321 &       0.03 &       0.01 &       0.04 &       0.01 &       0.01 &  146.10 &      3.281 &  1273.05 & 8489.750 \\
                 p\_hat300-2 &  300 & 0.49 &  31564 &  31564 &       0.15 &       6.44 &       6.59 &       6.96 &      42.90 & $>$1000 &    171.281 & $>$10800 & $>$10800 \\
                 p\_hat500-1 &  500 & 0.25 &   4764 &   4764 &       0.05 &       0.08 &       0.14 &       0.08 &       0.13 & $>$1000 &            &          &          \\
                 p\_hat700-1 &  700 & 0.25 &   5185 &   5185 &       0.08 &       0.37 &       0.45 &       0.38 &       0.52 & $>$1000 &            &          &          \\
                     san1000 & 1000 & 0.50 &  10661 &   6588 &       2.94 &      17.18 &      20.11 &      19.46 &    $>$1000 & $>$1000 &            &          &          \\
              san200\_0.7\_1 &  200 & 0.70 &  45295 &  45295 &       0.05 &       0.09 &       0.14 &       1.95 &      54.88 &   28.72 &            &          &          \\
              san200\_0.7\_2 &  200 & 0.70 &  15073 &  15073 &       0.12 &       2.34 &       2.46 &       4.02 &      17.86 & $>$1000 &            &          &          \\
              san200\_0.9\_1 &  200 & 0.90 & 242710 & 242710 &       0.04 &      38.58 &      38.62 &    $>$1000 &      12.56 &  206.01 &            &          &          \\
              san200\_0.9\_2 &  200 & 0.90 & 178468 & 178468 &       0.03 &      88.84 &      88.88 &     378.96 &     833.49 & $>$1000 &            &          &          \\
              san400\_0.5\_1 &  400 & 0.50 &   7442 &   7442 &       0.37 &       0.14 &       0.51 &       0.63 &      60.36 & $>$1000 &            &          &          \\
              san400\_0.7\_1 &  400 & 0.70 &  77719 &  77719 &       0.31 &      15.55 &      15.86 &     645.24 &    $>$1000 & $>$1000 &            &          &          \\
              san400\_0.7\_2 &  400 & 0.70 &  44155 &  44155 &       0.24 &      52.55 &      52.78 &     497.74 &    $>$1000 & $>$1000 &            &          &          \\
              san400\_0.7\_3 &  400 & 0.70 &  24727 &  24727 &       0.11 &     211.11 &     211.22 &     325.71 &    $>$1000 & $>$1000 &            &          &          \\
                sanr200\_0.7 &  200 & 0.70 &  16398 &  16398 &       0.01 &       4.14 &       4.16 &       4.34 &      18.67 & $>$1000 &            &          &          \\ \hline
        \end{tabular}
    \end{adjustbox}
\end{table*}

\begin{table*}[tb]
    \caption{Iterations for DIMACS}
    \label{tab:iterations for DIMACS}
    \centering
    \begin{adjustbox}{max width=\textwidth}
        \begin{tabular}{rrr|r|rr|rr|r}
            \hline
            &  &   & \multicolumn{3}{c|}{MECQ} & \multicolumn{2}{c|}{EWCLIQUE} & $\frac{\mbox{EWCLIQUE}}{\mbox{MECQ}}$ \\
            & & & with PLS & \multicolumn{2}{c|}{without PLS} & & & \multicolumn{1}{c}{iteration} \\
            graph & $|V|$ & $d$ & iterations & iterations & time [\si{\micro}s] & iterations & time [\si{\micro}s] & \multicolumn{1}{c}{ratio} \\ \hline
                 brock200\_1 &  200 & 0.75 &  6074449 &  6103600 &       3.95 & 1328614116 &       0.25 &   0.46\% \\
                 brock200\_2 &  200 & 0.50 &    14073 &    19906 &       3.01 &     345371 &       0.29 &   5.76\% \\
                 brock200\_3 &  200 & 0.61 &   114928 &   130560 &       3.22 &    4282305 &       0.30 &   3.05\% \\
                 brock200\_4 &  200 & 0.66 &   287037 &   310735 &       3.77 &   13814425 &       0.35 &   2.25\% \\
                      C125.9 &  125 & 0.90 &  4329351 &  4551897 &       5.45 &            &            &          \\
                  c-fat200-1 &  200 & 0.08 &       28 &       38 & $\epsilon$ &        632 & $\epsilon$ &   6.01\% \\
                  c-fat200-2 &  200 & 0.16 &       97 &      107 & $\epsilon$ &       6780 & $\epsilon$ &   1.58\% \\
                  c-fat200-5 &  200 & 0.43 &      113 &      141 & $\epsilon$ &  138193445 &       0.54 &   0.00\% \\
                 c-fat500-10 &  500 & 0.37 &     3853 &     3947 &      60.81 &            &            &          \\
                  c-fat500-1 &  500 & 0.04 &       61 &       66 & $\epsilon$ &       1605 & $\epsilon$ &   4.11\% \\
                  c-fat500-2 &  500 & 0.07 &       92 &      126 & $\epsilon$ &       4679 & $\epsilon$ &   2.69\% \\
                  c-fat500-5 &  500 & 0.19 &      324 &      404 &      24.75 &    1227023 &       0.35 &   0.03\% \\
                  DSJC500\_5 &  500 & 1.00 &  2419493 &  2494606 &       4.01 &  200152687 &       0.22 &   1.25\% \\
            gen200\_p0.9\_55 &  200 & 0.90 & 13443080 &          &            &            &            &          \\
                  hamming6-2 &   64 & 0.90 &       32 &       48 & $\epsilon$ &        896 & $\epsilon$ &   5.36\% \\
                  hamming6-4 &   64 & 0.35 &      265 &      265 & $\epsilon$ &        340 & $\epsilon$ &  77.94\% \\
                  hamming8-2 &  256 & 0.97 &   479056 &   479125 &      42.45 &      65731 &       3.50 & 728.92\% \\
                  hamming8-4 &  256 & 0.64 &    86597 &    88679 &       6.20 &    2475100 &       0.59 &   3.58\% \\
               johnson16-2-4 &  120 & 0.76 &   309697 &   309697 &       0.58 &    1905154 &       0.13 &  16.26\% \\
                johnson8-2-4 &   28 & 0.56 &       79 &       79 & $\epsilon$ &        150 & $\epsilon$ &  52.67\% \\
                johnson8-4-4 &   70 & 0.77 &      354 &      361 & $\epsilon$ &       3953 & $\epsilon$ &   9.13\% \\
                     keller4 &  171 & 0.65 &    61141 &    63170 &       3.32 &    2158496 &       0.32 &   2.93\% \\
                    MANN\_a9 &   45 & 0.93 &    35116 &    35128 &       0.57 &     116041 &       0.17 &  30.27\% \\
                p\_hat1000-1 & 1000 & 0.24 &   582124 &   591826 &       3.33 &    9890185 &       0.30 &   5.98\% \\
                p\_hat1500-1 & 1500 & 0.25 &  4552934 &  4565892 &       4.30 &  106284583 &       0.31 &   4.30\% \\
                 p\_hat300-1 &  300 & 0.24 &     3975 &     4221 &       2.37 &      50151 &       0.20 &   8.42\% \\
                 p\_hat300-2 &  300 & 0.49 &   876123 &  1053858 &       6.60 &  134486327 &       0.32 &   0.78\% \\
                 p\_hat500-1 &  500 & 0.25 &    27485 &    27601 &       2.90 &     468371 &       0.28 &   5.89\% \\
                 p\_hat700-1 &  700 & 0.25 &   110426 &   113403 &       3.35 &    1678557 &       0.31 &   6.76\% \\
                     san1000 & 1000 & 0.50 &   345909 &   383550 &      50.74 &            &            &          \\
              san200\_0.7\_1 &  200 & 0.70 &     6694 &   425248 &       4.59 &  387149894 &       0.14 &   0.11\% \\
              san200\_0.7\_2 &  200 & 0.70 &   335623 &   680897 &       5.90 &   48732878 &       0.37 &   1.40\% \\
              san200\_0.9\_1 &  200 & 0.90 &  1637404 &          &            &   12731307 &       0.99 &   0.00\% \\
              san200\_0.9\_2 &  200 & 0.90 &  4463309 & 25206475 &      72.79 &  303169816 &       2.75 &   1.72\% \\
              san400\_0.5\_1 &  400 & 0.50 &    11065 &    68967 &       9.13 &   43132933 &       1.40 &   0.16\% \\
              san400\_0.7\_1 &  400 & 0.70 &   547682 & 53869639 &     166.74 &            &            &          \\
              san400\_0.7\_2 &  400 & 0.70 &  2841349 & 57665379 &      64.93 &            &            &          \\
              san400\_0.7\_3 &  400 & 0.70 & 20591310 & 39873392 &      32.99 &            &            &          \\
                sanr200\_0.7 &  200 & 0.70 &  1045157 &  1196523 &       3.63 &   55871909 &       0.33 &   2.14\% \\ \hline
        \end{tabular}
    \end{adjustbox}
\end{table*}

\clearpage

\section{Conclusion}
\label{sec:Conclusion}

We proposed a branch-and-bound algorithm MECQ for MEWCP.
Our algorithm calculates upper bounds using vertex coloring.
In the vertex coloring procedure,
our algorithm assigns edge weights to vertices to calculate upper bounds.
By some computational experiments,
we confirmed our algorithm is faster than previous ones.

Although modern techniques are proposed for MCP \cite{li2017minimization,san2016improved},
they cannot be directly applied to MEWCP because of edge weights.
To apply such techniques to MEWCP, modifying them is a future work.

Recently, quantum annealer is studied to solve NP-hard problems including MCP
\cite{chapuis2018finding,bian2017solving}.
Quantum annealer can solve the quadratic unconstrained binary optimization (QUBO) problem.
Since quantum annealer solvers are heuristic, efficient exact solvers are required to evaluate them.
QUBO can be formulated as MEWCP
by the vertex-and-edge-weighted complete graphs where negative weight is allowed.
Hence handling negative weight is one future work.

\bibliographystyle{plain}
\bibliography{ref}

\end{document}